\documentclass[copyright,creativecommons]{eptcs}

\providecommand{\event}{ICLP 2023}
\usepackage{underscore}
\usepackage{amssymb,amsmath,amsthm}
\usepackage{enumitem}
\usepackage{graphics}

\title{Generalizing Level Ranking Constraints \\
for Monotone and Convex Aggregates}

\author{Tomi Janhunen
        \institute{Computing Sciences, Tampere University, Finland}
        \institute{Department of Computer Science, Aalto University, Finland}
	\email{tomi.janhunen@tuni.fi}}
\def\titlerunning{Generalizing Level Ranking Constraints for Monotone and Convex Aggregates}
\def\authorrunning{Tomi Janhunen}

\newtheorem{theorem}{Theorem}
\newtheorem{definition}{Definition}

\newtheorem{proposition}{Proposition}
\newtheorem{example}{Example}
\newcommand{\eofex}{\mbox{}\nobreak\hfill\hspace{0.5em}$\blacksquare$}

\newcommand{\union}{\cup}
\newcommand{\Union}{\bigcup}
\newcommand{\isect}{\cap}
\newcommand{\pair}[2]{\langle{#1},{#2}\rangle}
\newcommand{\tuple}[1]{\langle{#1}\rangle}
\newcommand{\rg}[3]{#1#2\dots#2#3}
\newcommand{\set}[1]{\{{#1}\}}
\newcommand{\eset}[2]{\{{#1},\ldots,{#2}\}}
\newcommand{\sel}[2]{\{{#1}\mid{#2}\}}
\newcommand{\compl}[1]{\overline{#1}}
\newcommand{\fdef}[3]{{#1}:{#2}\rightarrow{#3}}
\newcommand{\pset}[1]{\mathbf{2}^{#1}}
\newcommand{\nat}{\mathbb{N}}
\newcommand{\integers}{\mathbb{Z}}
\newcommand{\Land}{\bigwedge}
\newcommand{\Lor}{\bigvee}
\newcommand{\true}{\top}
\newcommand{\false}{\bot}
\newcommand{\lequiv}{\leftrightarrow}
\newcommand{\limpl}{\rightarrow}
\newcommand{\sig}[1]{\mathrm{At}(#1)}
\newcommand{\vars}{\mathcal{V}}
\newcommand{\heads}[1]{\mathrm{H}(#1)}
\newcommand{\sigv}[1]{\mathrm{At_v}(#1)}
\newcommand{\sigh}[1]{\mathrm{At_h}(#1)}

\DeclareMathSymbol{\naf}{\mathord}{symbols}{"18}
\newcommand{\IF}{\leftarrow}
\newcommand{\AND}{,}
\newcommand{\OR}{\mid}
\newcommand{\END}{.\;}
\newcommand{\GLred}[2]{#1^{#2}}
\newcommand{\wght}[1]{w_{#1}}
\newcommand{\wsum}[2]{\mathrm{WS}_{#1}(#2)}
\newcommand{\len}[1]{\|#1\|}
\newcommand{\sm}[1]{\mathrm{SM}(#1)}
\newcommand{\supp}[2]{\mathrm{SuppR}(#1,#2)}
\newcommand{\suppm}[1]{\mathrm{SuppM}(#1)}
\newcommand{\lm}[1]{\mathrm{LM}(#1)}
\newcommand{\mm}[2]{\mathrm{MM}_{#1}(#2)}
\newcommand{\system}[1]{\textsc{#1}}
\newcommand{\imt}[1]{\mathbf{T}_{#1}}
\newcommand{\iter}[3]{{#1}\uparrow^{#2}(#3)}
\newcommand{\level}[1]{\##1}
\newcommand{\pdep}[1]{\mathbin{\succeq}_{#1}}
\newcommand{\dgplus}[1]{\mathrm{DG}^+(#1)}
\newcommand{\aggr}[2]{\mathrm{Aggr}^{#1}(#2)}
\newcommand{\csubst}[3]{#1\,?\,#2\!:\!#3}
\newcommand{\tr}[3]{\mathrm{#1}^{#2}(#3)}
\newcommand{\trop}[2]{\mathrm{#1}^{#2}}
\newcommand{\defof}[2]{\mathrm{Def}_{#2}(#1)}
\newcommand{\app}[2]{\mathrm{app}^{#2}(#1)}
\renewcommand{\int}[2]{\mathrm{int}^{#2}(#1)}
\newcommand{\ext}[2]{\mathrm{ext}^{#2}(#1)}
\newcommand{\dep}[2]{\mathrm{dep}(#1,#2)}
\newcommand{\gap}[2]{\mathrm{gap}(#1,#2)}
\newcommand{\vub}[2]{\mathrm{vub}^{#2}(#1)}
\newcommand{\sccof}[1]{\mathrm{SCC}(#1)}
\newcommand{\body}[1]{\mathrm{B}(#1)}
\newcommand{\head}[1]{\mathrm{H}(#1)}
\newcommand{\pbody}[1]{\mathrm{B}^+(#1)}
\newcommand{\nbody}[1]{\mathrm{B}^-(#1)}
\newcommand{\lpeq}[1]{\equiv_{\mathrm{#1}}}

\newcommand{\citep}[1]{\cite{#1}}
\newcommand{\tconly}[1]{#1}

\begin{document}
\maketitle

\begin{abstract}
In answer set programming (ASP), answer sets capture solutions to
search problems of interest and thus the efficient computation of
answer sets is of utmost importance. One viable implementation
strategy is provided by translation-based ASP where logic programs
are translated into other KR formalisms such as Boolean
satisfiability (SAT), SAT modulo theories (SMT), and mixed-integer
programming (MIP). Consequently, existing solvers can be harnessed
for the computation of answer sets.
Many of the existing translations rely on program completion and level
rankings to capture the minimality of answer sets and default negation
properly. In this work, we take level ranking constraints into
reconsideration, aiming at their generalizations to cover
aggregate-based extensions of ASP in more systematic way.
By applying a number of program transformations, ranking constraints
can be rewritten in a general form that preserves the structure of
monotone and convex aggregates and thus offers a uniform basis for
their incorporation into translation-based ASP. The results open up
new possibilities for the implementation of translators and solver
pipelines in practice.
\end{abstract}

%------------------------------------------------------------------------------

\section{Introduction}

Answer set programming (ASP, \cite{BET11:cacm}) offers rich rule-based
languages for knowledge representation (KR) and reasoning. Given some
search or optimization problem of interest, its \emph{encoding} in ASP
is a logic program whose answer sets capture solutions to the
problem. Thus the efficient computation of answer sets is of utmost
importance. One viable implementation strategy is provided by
\emph{translation-based} ASP where logic programs are translated into
other KR formalisms such as
Boolean satisfiability (SAT, \cite{BHvMW21:handbook}),
SAT modulo theories (SMT, \cite{BSST21:faia}), or
mixed-integer programming (MIP, \cite{Wolsey08:book}).
Consequently, existing solver technology can be harnessed for the
computation of answer sets.

The semantics of answer set programming rests on \emph{stable models}
\citep{GL88:iclp} that incorporate a notion of minimality and give a
declarative semantics for default negation. Capturing these aspects in
satisfaction-based formalisms such as pure SAT is non-trivial; see,
e.g., \citep{Janhunen06:jancl,LR06:tocl}.
There are also various syntactic aggregations \citep{AFG23:tplp} that
enable compact encodings but whose translation is potentially expensive
if there is no respective primitive in the target formalism.
A typical translation consists of several steps such as
(i) \emph{normalization} \citep{BJN20:acmtocl},
(ii) \emph{instrumentation} for loop prevention
\citep{BGJKS16:fi,Janhunen06:jancl,LZ04:aij}, and
(iii) \emph{completion} \citep{Clark78:ldb}.
The first step concerns the removal of syntactic extensions that have
been introduced to increase the expressive power of ASP in favor of
\emph{normal} rules. The second step either transforms the program or
adds suitable constraints so that the difference between stable and
\emph{supported models} disappears. The third step captures supported
models by transforming rules into equivalences.
Ideally, the syntactic details of the target language are deferred
during translation and incorporated at the very end; either after or
while forming the completion. This strategy realizes a
\emph{cross-translation} approach \citep{Janhunen18:ki} in analogy to
modern compiler designs.

Many of the existing translations
\citep{JNS09:lpnmr,LJN12:kr,NJN11:inap}
essentially rely on \emph{level ranking constraints}
formalized by \tconly{Niemel\"a} \cite{Niemela08:amai} as formulas in
difference logic \citep{NOT06:jacm}. Such constraints describe
\emph{level numbers} that order the atoms of a normal program in such
a way that stable models can be distinguished among the supported ones
\citep{Fages94:jmlcs}. Thus, level numbers are essential when it comes
to capturing the minimality of stable models and the semantics of
default negation properly.
As shown in \citep{Janhunen06:jancl}, level numbers can be made unique
so that they match with the levels of atoms obtained by applying the
\emph{immediately true} operator $\imt{P}$ iteratively. Uniqueness can
also be enforced in terms of \emph{strong} level ranking constraints
\citep{Niemela08:amai}. Unique level numbers are also highly desirable
when aiming at one-to-one correspondences with stable models, e.g.,
when counting solutions to problems or carrying out probabilistic
inference \citep{FBRSGTJR15:tplp}.

In this work, we take level ranking constraints into reconsideration,
aiming at generalizations that cover aggregate-based extensions of ASP
in a more systematic way. So far, only normal programs are truly covered
\citep{Janhunen06:jancl,Niemela08:amai}
and the normalization of input programs is presumed. The
generalization for weight constraint programs (WCPs), as sketched by
\tconly{Liu et al.}~\cite{LJN12:kr}, concerns only weak constraints
and is confined to translations into MIP.
However, the idea of avoiding or delaying normalization is interesting
as such, opening up new possibilities for ordering the translation
steps discussed above. For instance, if $\tr{NORM}{}{\cdot}$ and
$\tr{LRC}{}{\cdot}$ stand for translations based on normalization and
level ranking constraints, respectively, it would be highly
interesting to compare $\tr{LRC}{}{\tr{NORM}{}{P}}$ with potential
generalizations $\tr{LRC}{}{P}$ that express level ranking
constraints at an aggregated level. Such representations are expected
to be more compact and to favor level rankings with fewer
variables. The resulting formulas can also be \emph{Booleanized}
afterwards \citep{Huang08:cp}, if translations toward SAT are
desirable, or rewritten in some other form that complies with the
intended back-end formalism.
In the sequel, we use translations $\tr{LRC}{}{\tr{NORM}{}{P}}$ of
increasingly complex programs $P$ as guidelines when lifting level
ranking constraints for aggregates. The idea is to cover program
classes involving standard aggregations subject to recursion. It turns
out that the structure of monotone and convex aggregates can be
preserved to a high degree, offering a uniform basis for their
incorporation into translation-based ASP.
On the one hand, the resulting generalizations exploit
\emph{ordered completion} \citep{ACZZ15:aij} in the reformulation
of \emph{weak} level ranking constraints but, on the other hand,
make novel contribution when imposing the uniqueness of level
rankings with \emph{strong} ones.

The rest of this article is structured as follows.
We begin by recalling the basic notions of \emph{logic programs} in
Section~\ref{section:preliminaries}, including the usual syntactic
fragments, stable and supported model semantics, and other concepts
relevant for this work.
Then, in Section~\ref{section:normal-ranking-constraints}, we explain
the details of ranking constraints in their standard form
corresponding to \emph{normal} logic programs. Actually, we present
them in a slightly rewritten form in order to pave the way for their
generalization for monotone aggregates in
Section~\ref{section:monotone-aggregates}.
Therein, we begin the analysis from the case of (positive) cardinality
and weight rules and, eventually, incorporate negative conditions to
ranking constraints.
To illustrate the generality of the results even further, we investigate
how certain convex aggregates get also covered via appropriate
program transformations in Section~\ref{section:convex-aggregates}.
Finally, the conclusions of this work are presented in
Section~\ref{section:conclusions}.

\section{Preliminaries}
\label{section:preliminaries}

In the sequel, we will consider \emph{propositional logic programs} that
are finite sets rules of the forms
\eqref{eq:normal-rule}--\eqref{eq:disjunctive-rule}
below. In the rules, $a$, $a_i$'s, $b_j$'s, and $c_k$'s are
\emph{propositional atoms} (or \emph{atoms} for short) and $\naf$
denotes {\em default negation}. The rules in
\eqref{eq:normal-rule}--\eqref{eq:disjunctive-rule}
are known as \emph{normal}, \emph{choice}, \emph{cardinality}, 
\emph{weight}, and \emph{disjunctive} rules, respectively.
Each rule
has a \emph{head} and a \emph{body} separated by the $\IF$ sign, and
the rough intuition is that if the \emph{condition(s)} formed by the
rule body are satisfied, then the respective head atom $a$ in
\eqref{eq:normal-rule}--\eqref{eq:weight-rule}, or some of the head
atoms $\rg{a_1}{,}{a_h}$ in \eqref{eq:disjunctive-rule} can be derived.

\begin{eqnarray}
\label{eq:normal-rule}
a & \IF & \rg{b_1}{\AND}{b_n}\AND\rg{\naf c_1}{\AND}{\naf c_m}\END
\\
\label{eq:choice-rule} % NB: Partitioning counts on single-head choice rules
\set{a}
& \IF & \rg{b_1}{\AND}{b_n}\AND\rg{\naf c_1}{\AND}{\naf c_m}\END
\\
\label{eq:cardinality-rule}
a & \IF & l\leq\set{\rg{b_1}{,}{b_n},\rg{\naf c_1}{,}{\naf c_m}}\END
\\
\label{eq:weight-rule}
a & \IF & w\leq\set{\rg{b_1=\wght{b_1}}{,}{b_n=\wght{b_n}},
                    \rg{\naf c_1=\wght{c_1}}{,}{\naf c_m=\wght{c_m}}}\END
\\
\label{eq:disjunctive-rule}
\rg{a_1}{\OR}{a_h}
& \IF & \rg{b_1}{\AND}{b_n}\AND\rg{\naf c_1}{\AND}{\naf c_m}\END
\end{eqnarray}
The choice regarding the head of \eqref{eq:choice-rule} is optional
while \eqref{eq:disjunctive-rule} insists on deriving at least one
head atom $a_i$ in a minimal way, as to be detailed in
Definition~\ref{def:stable-model}.
A positive body condition $b_j$ holds if $b_j$ can be derived by
some other rules whereas $\naf c_k$ holds if $c_k$ cannot be derived.
A cardinality rule \eqref{eq:cardinality-rule} demands that at least $l$
of such conditions are met to activate the rule. Weight rules
\eqref{eq:weight-rule} are similar but body conditions $b_j$ and $\naf
c_k$ are valued by their respective non-negative integer weights
$\wght{b_j}$ and $\wght{c_k}$ when it comes to reaching the bound $w$.
In the sequel, we use shorthands
$\pbody{r}=\eset{b_1}{b_n}$,
$\nbody{r}=\eset{c_1}{c_m}$, and
$\body{r}=\eset{b_1}{b_n}\union\eset{\naf c_1}{\naf c_m}$
when referring to the body conditions occurring in a rule $r$. The set
of head atoms in $r$ is denoted by $\head{r}$ and for entire programs
$P$, we define $\heads{P}=\Union_{r\in P}\head{r}$.

Typical (syntactic) classes of logic programs are as follows:
\emph{normal} logic programs (NLPs) consist of normal rules
\eqref{eq:normal-rule} and the same can be stated about
\emph{disjunctive} logic programs (DLPs) and disjunctive rules
\eqref{eq:disjunctive-rule} that are normal as a special case ($h=1$).
The class of \emph{weight constraint programs} (WCPs)
\citep{SNS02:aij} is essentially based on normal rules
\eqref{eq:normal-rule} and the aggregated rule types in
\eqref{eq:choice-rule}--\eqref{eq:weight-rule}, out of which weight
rules are expressive enough to represent the class of WCPs alone.
Contemporary ASP systems---aligned with the ASP-core-2 language standard
\citep{CFGIKKLRS13:aspcore2}---support these fragments so well that
programmers can mix rule types freely in their encodings.
When the fragment is not important, we may refer to
\emph{logic programs} or \emph{programs} for short.
Finally, we say that a rule is \emph{positive}%
\footnote{Note that the head $\set{a}$ of choice rule embeds hidden
(double) negation since it can be expressed as $a\IF\naf\naf a$.}
if $m=0$ and it is of the forms \eqref{eq:normal-rule},
or \eqref{eq:cardinality-rule}--\eqref{eq:disjunctive-rule}.
An entire program is called positive if its rules are all positive.

The \emph{signature} $\sig{P}$ of a logic program $P$ is the
set of atoms that occur in the rules of $P$.
An {\em interpretation} $I\subseteq\sig{P}$ of $P$ tells which atoms
$a\in\sig{P}$ are \emph{true} ($a\in I$, also denoted $I\models a$) whereas
others are \emph{false} ($a\in\sig{P}\setminus I$, denoted $I\not\models
a$). Atoms are also called \emph{positive literals}. Any
\emph{negative literal} $\naf c$, where $c$ is an atom, is understood
classically, i.e., $I\models\naf c$, iff $I\not\models c$.
The relation $\models$ extends for the bodies of
normal/choice/disjunctive rules $r$ as follows:
$I\models\body{r}$
iff $\pbody{r}\subseteq I$ and $\nbody{r}\isect I=\emptyset$.
The body $l\leq\body{r}$ of a cardinality rule $r$ is satisfied in $I$
iff $l\leq |\pbody{r}\isect I|+|\nbody{r}\setminus I|$. More generally,
the body of a weight rule $r$ in \eqref{eq:weight-rule} is satisfied
in $I$ iff the \emph{weight sum}
$
\wsum{I}{\rg{b_1=\wght{b_1}}{,}{b_n=\wght{b_n}},
          \rg{\naf c_1=\wght{c_1}}{,}{\naf c_m=\wght{c_m}}}
$
$=$
$\sum_{b\in\pbody{r}\isect I}\wght{b}+\sum_{c\in\nbody{r}\setminus I}\wght{c}$
is at least $w$. For rules $r$, we have $I\models r$ iff the
satisfaction of its body implies the satisfaction of its head, except
that choice rules \eqref{eq:choice-rule} are always satisfied.
An interpretation $I\subseteq\sig{P}$ is a (\emph{classical})
\emph{model} of a program $P$, denoted $I\models P$, iff $I\models r$
for each $r\in P$.
A model $M\models P$ is $\subseteq$-\emph{minimal} iff there is no
$M'\models P$ such that $M'\subset M$. The set of $\subseteq$-minimal
models of $P$ is denoted by $\mm{}{P}$.  If $P$ is positive and
non-disjunctive then $|\mm{}{P}|=1$ and the respective
\emph{least model} of $P$ is denoted by $\lm{P}$.

\begin{definition}[Stable models \citep{GL88:iclp,GL91:ngc,SNS02:aij}]
\label{def:reduct}\label{def:stable-model}
For a program $P$ and an interpretation $I\subseteq\sig{P}$, the
\emph{reduct} $\GLred{P}{I}$ of $P$ with respect to $I$ contains
\begin{enumerate}
\item
a rule $a\IF\pbody{r}$ for each normal rule \eqref{eq:normal-rule}
such that $\nbody{r}\isect I=\emptyset$, and for each choice rule
\eqref{eq:choice-rule} such that $a\in I$ and $\nbody{r}\isect I=\emptyset$;

\item
a rule $a\IF l'\leq\pbody{r}$ for each cardinality rule
\eqref{eq:cardinality-rule} and the bound $l'=\max(0,l-|\nbody{r}\setminus I|)$;

\item
a rule $a\IF w'\leq\set{\rg{b_1=\wght{b_1}}{,}{b_n=\wght{b_n}}}$ for each
weight rule \eqref{eq:weight-rule} and the bound
$w'=\max(0,w-\wsum{I}{\rg{\naf c_1=\wght{c_1}}{,}{\naf c_m=\wght{c_m}}})$; and

\item
a rule $\rg{a_1}{\OR}{a_h}\IF\pbody{r}$ for each disjunctive rule
\eqref{eq:disjunctive-rule} such that $\nbody{r}\isect I=\emptyset$.
\end{enumerate}
An interpretation $M\subseteq\sig{P}$ is a~\emph{stable model} of
the program $P$ iff $M\in\mm{}{\GLred{P}{M}}$.
\end{definition}

\begin{example}
\label{ex:positive-card}
Consider a cardinality rule $a\IF 1\leq\set{\rg{b_1}{\AND}{b_n}}$
with choice rules $\rg{\set{b_1}\END}{}{\set{b_n}\END}$
Besides the empty stable model $\emptyset$, these rules induce $2^n-1$
stable models $M=\set{a}\union N$ with $\emptyset\subset
N\subseteq\set{\rg{b_1}{,}{b_n}}$: the head $a$ is set true whenever
\emph{at least} one of $\rg{b_1}{,}{b_n}$ is chosen to be true.
\eofex
\end{example}

In the sequel, we mostly concentrate on non-disjunctive programs
$P$. Then, the stability of $M\subseteq\sig{P}$ can also be captured
with the fixed point equation $M=\lm{\GLred{P}{M}}$. Moreover, the
well-known $\imt{P}$ operator, when applied to an interpretation
$I\subseteq\sig{P}$, produces the set of atoms $a\in\sig{P}$ that are
\emph{immediately true} under $I$, i.e., for which there is a positive
rule $r$ having $a$ as the head and whose body is satisfied by $I$.
It follows that $M\models P$ holds for a non-disjunctive program $P$
iff $\imt{\GLred{P}{M}}(M)\subseteq M$ and $M$ is a
\emph{supported model} of $P$ iff $M=\imt{\GLred{P}{M}}(M)$.
Given a supported model $M$, the support is provided by the set of
rules $\supp{P}{M}\subseteq P$ whose bodies are satisfied by $M$.
Since $\lm{\GLred{P}{M}}$ is obtained as the least fixed point
$\iter{\imt{\GLred{P}{M}}}{\infty}{\emptyset}$, each stable
model of $P$ is also supported.
We write $\sm{P}$ and $\suppm{P}$ for the sets of stable and supported
models of $P$, respectively. Thus $\sm{P}\subseteq\suppm{P}$ holds
in general.

Next we recall some concepts related to modularity. First, given a WCP
$P$, the set of \emph{defining rules} for an atom $a\in\heads{P}$ is
$\defof{a}{P}=\sel{r\in P}{a\in\head{r}}$.  Thus $P$ can be
partitioned as $\Union_{a\in\heads{P}}\defof{a}{P}$.
Second, the \emph{positive dependency graph} of $P$ is
$\dgplus{P}=\tuple{\sig{P},\pdep{P}}$ where $a\pdep{P} b$ holds for
$a,b\in\sig{P}$, if $a\in\head{r}$ and $b\in\pbody{r}$ for some rule
$r\in\defof{a}{P}$. A \emph{strongly connected component} (SCC) of
$\dgplus{P}$ is a maximal subset $S\subseteq\sig{P}$ such that all
distinct atoms $a,b\in S$ depend on each other via directed paths in
$\dgplus{P}$. For an atom $a\in\heads{P}$, the SCC of $a$ is denoted
by $\sccof{a}$.
As shown in \citep{OJ08:tplp}, each SCC $S$ of a WCP $P$ gives rise to
a \emph{program module} $P_S=\Union_{a\in S}\defof{a}{P}$ where pure
body atoms $b\in\sig{P_S}\setminus S$ are treated as \emph{input atoms}
taking any truth value, intuitively defined by choice rules
$\set{b}$. This yields a set of stable models $\sm{P_S}$ for
each module $P_S$ based on Definition \ref{def:stable-model}.
Given two stable models $M\in\sm{P}$ and $N\in\sm{Q}$, we say that $M$
and $N$ are mutually \emph{compatible}, if they agree on the truth
values of atoms in $\sig{P}\isect\sig{Q}$, i.e.,
$M\isect\sig{Q}=N\isect\sig{P}$. The \emph{module theorem} of
\cite{OJ08:tplp} states that the stable models of $P$ can be obtained
as mutually compatible collections of stable models
$\rg{M_1}{,}{M_n}$ for the program modules $\rg{P_{S_1}}{,}{P_{S_n}}$
induced by the SCCs $\rg{S_1}{,}{S_n}$ of $P$.

Finally, some notions of equivalence should be introduced. Logic
programs $P$ and $Q$ are \emph{weakly equivalent}, denoted
$P\lpeq{}Q$, iff $\sm{P}=\sm{Q}$. They are \emph{strongly equivalent},
denoted $P\lpeq{s}Q$, iff $P\union R\lpeq{}Q\union R$ for any other
context program $R$ \citep{LPV01:acmtocl}. Then $P\lpeq{s}Q$ implies
$P\lpeq{}Q$ but not vice versa. Strong equivalence can be characterized
by using only contexts formed by \emph{unary} positive rules $a\IF b$, or
semantically by using SE-models \citep{Turner03:tplp}.
To address the correctness of various translations, however, more
fine-grained relations become necessary. The signature $\sig{P}$ of a
logic program $P$ can be split into \emph{visible} and \emph{hidden}
parts $\sigv{P}$ and $\sigh{P}$, respectively.  Given a stable model
$M\in\sm{P}$, only its visible projection $M\isect\sigv{P}$ is
relevant when comparing $P$ with other programs. Thus, $P$ and $Q$ are
\emph{visibly equivalent}, denoted $P\lpeq{v}Q$, iff
$\sigv{P}=\sigv{Q}$ and $M\isect\sigv{P}=N\isect\sigv{Q}$ holds for
each pair of models $M\in\sm{P}$ and $N\in\sm{Q}$ in a bijective
correspondence \citep{Janhunen06:jancl}. There is a generalization of
both $\lpeq{v}$ and $\lpeq{s}$, viz.~\emph{visible strong equivalence}
$\lpeq{vs}$, that incorporates context programs $R$ that
\emph{respect the hidden atoms} of $P$ and $Q$
for comparisons \citep{BJN20:acmtocl}. The correctness of normalization
has been addressed in this sense. E.g., for a weight rule $r$ in
\eqref{eq:weight-rule}, $\set{r}\lpeq{vs}\tr{NORM}{}{\set{r}}$, which
meas that $r$ can be safely substituted by $\tr{NORM}{}{\set{r}}$ in
contexts respecting the hidden atoms introduced by the normalization.

\section{Level Rankings and Ranking Constraints}
\label{section:normal-ranking-constraints}

When a stable model $M\subseteq\sig{P}$ of a \emph{non-disjunctive}
logic program $P$ is constructed using the reduct $\GLred{P}{M}$ and
the $\imt{\GLred{P}{M}}$ operator, atoms true in the model $M$ get
divided into \emph{levels}
$M_i=(\iter{\imt{\GLred{P}{M}}}{i}{I})
      \setminus(\iter{\imt{\GLred{P}{M}}}{i-1}{I})$
where $i>0$ and $I\subseteq\sig{P}\setminus\heads{P}$ is a set of
input atoms. By default $I=\emptyset$ and $M_0=\emptyset$,
but if $I\neq\emptyset$, then $M_0=I$.
For finite programs $P$, the index $i$ is bounded from above by
$|\sig{P}|$. Based on this division of atoms, it is possible to read
off a \emph{level ranking}
$\fdef{\level{}}{\sig{P}}{\nat\union\set{\infty}}$ for the atoms of
the program \citep{Niemela08:amai}:
the rank $\level{a}=i$, if $a\in M_i$,
and $\level{a}=\infty$, if $a\not\in M$.
A \emph{level numbering} $\level{}$ \citep{Janhunen06:jancl} extends
any level ranking for the supporting rules $r\in\supp{P}{M}$ by the
equality%
\footnote{This holds for rules whose body is essentially normal
\eqref{eq:normal-rule} while generalizations for more complex bodies follow.}
$\level{r}=\max\sel{\level{b}}{b\in\pbody{r}}+1$.
Intuitively, the level $\level{r}$ of a rule $r$ indicates when $r$
can be applied to derive its head and, consequently,
$\level{a}=\min\sel{\level{r}}{r\in\defof{a}{P}\isect\supp{P}{M}}$.
By these interconnections, we may use level rankings and numberings
interchangeably in the sequel. If $r\not\in\supp{P}{M}$, then
$\level{r}=\infty$. The value $\infty$ emphasizes that an atom is
never derived or a rule becomes never applicable. The other option is
to restrict the domain of $\level{}$ to $M\union\supp{P}{M}$ for which
finite values exist, but some big value greater than any level rank is
useful in practice. E.g., given an SCC $S$ of the program $P$, the
level ranks $\level{a}$ of atoms $a\in S$ can be effectively
constrained by $0<\level{a}<|S|+1$; cf.~\eqref{eq:levels} below.

Many existing translations of logic programs into SAT, SMT, and MIP
rely on program completion \citep{Clark78:ldb}. The idea is to
translate a (normal) logic program $P$ into classical equivalences
that capture the \emph{supported models} of the program.
The purpose of level ranking constraints \citep{Niemela08:amai},
however, is to distinguish the stable ones among them by incorporating
a requirement that there is a level ranking $\level{}$ for a model
$M\in\suppm{P}$. These constraints can be expressed, e.g., as formulas
in \emph{difference logic} (DL). This SMT-style logic
\citep{NOT06:jacm} enriches propositional formulas with difference
constraints of the form $x-y\leq k$ where $x,y$ are real/integer
variables and $k$ is a constant.
The evaluation of a difference atom $x-y\leq k$ is based on an
assignment $\fdef{\tau}{\vars}{\integers}$ on the set of variables
$\vars$ in use. Given $\tau$, the constraint $x-y\leq k$ is satisfied
by $\tau$, denoted by $\tau\models x-y\leq k$, iff
$\tau(x)-\tau(y)\leq k$. A \emph{DL-interpretation} is a pair
$\pair{I}{\tau}$ where $I$ a standard propositional interpretation and
$\tau$ an assignment. A formula $\phi$ of DL is satisfied by
$\pair{I}{\tau}$, denoted $\pair{I}{\tau}\models\phi$, if $\phi$
evaluates to true under $I$ by the usual propositional rules extended
by the evaluation of difference atoms subject to $\tau$.

A difference constraint $x_b-x_a\leq -1$ (i.e., $x_a>x_b$) can express
that $a$ is derived \emph{after} $b$ under the assumption that $x_a$
and $x_b$ store the level ranks of $a$ and $b$, respectively. Based on
this idea, we introduce formulas for the representation of level
ranks. Their \emph{scope} is specified in terms of a set of atoms
$S\subseteq\sig{P}$ to be discussed in further detail below.
\begin{eqnarray}
(1\leq x_a\leq |S|+1), & ~~ & \neg a\limpl(x_a\geq|S|+1), \label{eq:levels} \\
\dep{a}{b}   & \lequiv & b\land(x_a>x_b), \label{eq:dep} \\
\gap{a}{b}   & \lequiv & b\land(x_a>x_b+1). \label{eq:depdep}
\end{eqnarray}
By the two formulas in \eqref{eq:levels}, level ranks are positive and
fixed to $|S|+1$ if an atom $a$ is false. In addition, we introduce
two kinds of new atoms to help with the formulation of the actual
level ranking constraints.
First, the atom $\dep{a}{b}$ defined by \eqref{eq:dep}, denotes an
\emph{active} dependency of a head atom $a$ on a positive body atom
$b$, i.e., $b$ must be true. Such dependencies are deployed by
\tconly{Bomanson et al.}~\cite{BGJKS16:fi},
but we use a definition in terms of the difference constraint.
Second, the atom $\gap{a}{b}$, as defined by \eqref{eq:depdep}, means
a similar relationship except that $b$ is derived so early that it is
not critical for determining the exact level rank of $a$. Note that
$\gap{a}{b}$ implies $\dep{a}{b}$ in general but not vice versa. In
particular, if $\dep{a}{b}$ is true and $\gap{a}{b}$ is false, then
$b$ must be true and $x_a=x_b+1$, indicating that $a$ is derived right
after $b$. Such body atoms $b$ from the preceding level are relevant
when $a$ is derived by some rule $r\in\defof{a}{P}$ at level $x_a$.

In the following, we present a reformulation of level ranking
constraints \citep{Niemela08:amai} by exploiting the dependency
relations from \eqref{eq:dep} and \eqref{eq:depdep}.
Our further goal is to incorporate the idea of
\emph{ordered completion} \citep{ACZZ15:aij} for the sake
of more compact representation.
Given an atom $a\in\sig{P}$, its completion is based on the set
$\defof{a}{P}=\eset{r_1}{r_k}$ of its defining rules. In the sequel,
the \emph{applicability} of a rule $r_i$ is denoted by a new atom
$\app{r_i}{}$.
\begin{eqnarray}
a & \lequiv & \rg{\app{r_1}{}}{\lor}{\app{r_k}{}},
\label{eq:def-of-atom}
\\
\app{r_i}{} & \lequiv &
\Land_{b\in\pbody{r_i}\isect S}\dep{a}{b}
\land(\pbody{r_i}\setminus S)\land\neg\nbody{r_i}
~~~~~~ (1\leq i\leq k),
\label{eq:def-of-body}
\\
\app{r_i}{} & \limpl & \Lor_{b\in\pbody{r_i}\isect S}\neg\gap{a}{b}
~~~~~~ (1\leq i\leq k, ~~ \pbody{r_i}\isect S\neq\emptyset),
\label{eq:deny-gaps}
\\
\app{r_i}{} & \limpl & (x_a\leq 1)
~~~~~~ (1\leq i\leq k, ~~ \pbody{r_i}\isect S=\emptyset).
\label{eq:reset-ext}
\end{eqnarray}

Intuitively, the equivalence \eqref{eq:def-of-atom} sets the head atom $a$
true if and only if at least one of its defining rules $r_i$ is applicable.
This, in turn, is defined by the equivalence \eqref{eq:def-of-body}
insisting that atoms in $\pbody{r_i}\isect S$ have been previously
derived and all remaining positive and negative body conditions are
satisfied. This formulation embeds both
\emph{weak} level ranking constraints \citep{Niemela08:amai}
and ordered completion \citep{ACZZ15:aij} but relative to the set $S$.
The constraint \eqref{eq:deny-gaps} is the counterpart of
\emph{strong} level ranking constraints \citep{Niemela08:amai}
enforcing the minimality of level ranks assigned to atoms. Besides
this, the formula \eqref{eq:reset-ext} resets the level of the head
atom $a$ to $1$ when $a$ can be derived by applying an
\emph{externally supporting} rule $r_i$
with $\pbody{r_i}\isect S=\emptyset$.

Regarding the scope $S$, it is natural to think that the head atom $a$
is included usually. Also, few special cases deserve further attention.
(i)
If $S=\sig{P}$, then the completion becomes fully ordered, i.e.,
$\pbody{r_i}\setminus S$ becomes empty in \eqref{eq:def-of-body} and
the formula \eqref{eq:reset-ext} is generated only for
$r_i\in\defof{a}{P}$ with an empty $\pbody{r_i}$.
Moreover, if all atoms of $\sig{P}$ are completed using
\eqref{eq:def-of-atom}--\eqref{eq:reset-ext}, the resulting formulas
capture stable models directly, including level ranks of atoms.
(ii)
If $S=\sccof{a}$, then the ordering becomes local to the component
$S$. Then, if all atoms of $S$ are completed, the formulas capture
stable models $M$ for the \emph{program module} $P_S$ induced by the
component $S$ \citep{OJ08:tplp}. It should be emphasized that the
input atoms in $\sig{P_S}\setminus S$ are not subject to completion
and they may vary freely. Therefore, given a set of facts
$I\subseteq\sig{P_S}\setminus S$ as an actual \emph{input} for $P_S$,
the stable models of $P_S$ become solutions to
$M=\lm{\GLred{P_S}{M}\union I}$ whose levels $i$ are determined by
$\iter{\imt{\GLred{P_S}{M}}}{i}{I}$.
(iii)
Finally, if $S=\emptyset$ and $a\not\in S$ as an exception, equations
\eqref{eq:def-of-atom} and \eqref{eq:def-of-body} capture the standard
completion of $a$, the formula \eqref{eq:deny-gaps} becomes void, and
the formula \eqref{eq:reset-ext} ensures that $x_a=1$ whenever $a$ is true.

\begin{example}
As a minimal example, consider $a\IF a$ as the only rule $r_1$ of a
program $P$ and the SCC $S=\set{a}=\sccof{a}$.
We obtain the following formulas:
$(1\leq x_a\leq 2)$,
$\neg a\limpl(x_a\geq 2)$,
$\dep{a}{a}\limpl a\land(x_a>x_a)$,
$\gap{a}{a}\limpl a\land(x_a>x_a+1)$,
$a\lequiv\app{r_1}{}$,
$\app{r_1}{}\lequiv\dep{a}{a}$,
$\app{r_1}{}\limpl\neg\gap{a}{a}$.
They can be satisfied by falsifying $a$ and all new atoms, as well as
by setting $x_a=2$, indicating that $M_1=\emptyset$ is stable.
On the other hand, $M_2=\set{a}$ is not stable, which can be realized
by an attempt to make $a$ true in the formulas listed above. Thus
$\app{r_1}{}$ and $\dep{a}{a}$ must be true, too, and $\gap{a}{a}$
false. By further inspection of the formulas, it follows that
$x_a>x_a$ is true and $x_a>x_a+1$ is false, both indicating a
contradiction.
\eofex
\end{example}

The case $S=\sccof{a}$ is the most general one and deserves
justifications for correctness due to reformulations done in view
of \citep{Niemela08:amai} and the limitations of ordered completion
\citep{ACZZ15:aij} with regard to \eqref{eq:deny-gaps}.

% See \ref{section:appendix} for a proof of Theorem
% \ref{theorem:normal-correctness}.

\begin{definition}\label{def:TOC}
Given a normal logic program $P$ and a scope $S\subseteq\sig{P}$
of completion, the \emph{tight ordered completion} (TOC) of $P$
\emph{relative} to $S$ is the set of formulas 
\eqref{eq:levels} for $a\in S$,
\eqref{eq:dep} and \eqref{eq:depdep} for $a,b\in S$ whenever $a\pdep{P}b$,
and
\eqref{eq:def-of-atom}--\eqref{eq:reset-ext}
for each $a\in\sig{P}$ and $r_i\in\defof{a}{P}$.
\end{definition}

The TOC of $P$ relative to $S$ is denoted by $\tr{TOC}{S}{P}$ and we
omit $S$ from the notation $\tr{TOC}{S}{P}$, if $S=\sig{P}$. It is
worth noting that the length $\len{\tr{TOC}{S}{P}}$ stays linear in
$\len{P_s}$.

\begin{theorem}\label{theorem:normal-correctness}
Let $P$ be a normal logic program, $S$ an SCC of $P$,
and $P_S$ the module of $P$ induced by $S$.
\begin{enumerate}
\item
If $M\subseteq\sig{P_S}$ is a stable model of $P_S$ for an input
$I\subseteq\sig{P_S}\setminus S$ and $\fdef{\level{}}{M\isect S}{\nat}$
the respective level ranking, then there is a model $\pair{N}{\tau}$
for $\tr{TOC}{S}{P}$ such that
$M=N\isect\sig{P_S}$,
$\tau(x_a)=\level{a}$ for each $a\in M\isect S$,
and $\tau(x_a)=|S|+1$ for each $a\in S\setminus M$.

\item
If $\pair{N}{\tau}$ is a model of $\tr{TOC}{S}{P}$, then
$M=N\isect\sig{P_S}$ is a stable model of $P_S$ for the input
$I=N\isect(\sig{P_S}\setminus S)$ and for each $a\in M\isect S$,
$\level{a}=\tau(x_a)-\tau(z)$ is the respective level rank.
\end{enumerate}
\end{theorem}

\noindent
As a preparatory step toward generalizations for aggregated rules, our
final example in this section illustrates $\trop{TOC}{S}$ in the
context of a cardinality rule \eqref{eq:cardinality-rule} that is
normalized before completion.

\begin{example}\label{ex:card-TOC}
Let us assume that an atom $a$ is defined by a single cardinality rule
$a\IF 1\leq\set{\rg{b_1}{\AND}{b_n}}$
as part of a larger program $P$ having an SCC $S=\sccof{a}$ such
that $\eset{b_1}{b_n}\subseteq S$.
The rule is compactly expressible even without auxiliary atoms
in terms of $n$ positive normal rules 
\begin{center}
$\rg{a\IF b_1\END}{~}{a\IF b_n\END}$
\end{center}
The tight ordered completion produces the following formulas for
the joint head atom $a\in S$:
\begin{eqnarray}
a\lequiv\rg{\app{a}{1}}{\lor}{\app{a}{n}},
\label{eq:ex-def-of-atom}
\\
\rg{\app{a}{1}\lequiv\dep{a}{b_1}~}{,}{~\app{a}{n}\lequiv\dep{a}{b_n}},
\label{eq:ex-def-of-body}
\\
\rg{\app{a}{1}\limpl\neg\gap{a}{b_1}~}{,}{~\app{a}{n}\limpl\neg\gap{a}{b_n}}.
\label{eq:ex-deny-gaps}
\end{eqnarray}
In the above, we adopt a convention that $\app{a}{i}$ denotes the
application of $r_i\in\defof{a}{P}$.  Since each $b_i\in S$ the
respective rules $a\IF b_i$ may not contribute to external support via
\eqref{eq:reset-ext}.
\eofex
\end{example}

\section{The  Case of Monotone Aggregates}
\label{section:monotone-aggregates}

Cardinality rules \eqref{eq:cardinality-rule} and weight rules
\eqref{eq:weight-rule} with \emph{lower bounds} are widely used
examples of monotone aggregates and, in particular, if the
(anti-monotone) effect of negative literals is disregarded in the
sense of stable models (cf. Definition \ref{def:reduct}). The level
number $\level{a}$ of an atom $a\in\sig{P}$ is generalized in a
straightforward way when \emph{positive} cardinality/weight rules are
incorporated to the definition of the $\imt{P}$ operator
\citep{SNS02:aij}. As before, $\level{a}$ is the least value
$i\in\nat$ such that $a\in\iter{\imt{P}}{i}{\emptyset}$ for positive
programs $P$. Default negation is analogously treated via the reduct,
i.e., given a stable model $M\subseteq\sig{P}$, the operator
$\imt{\GLred{P}{M}}$ can be used to assign level ranks for
$a\in\sig{P}$.
The goal of this section is to generalize tight ordered completion for
rules involving monotone aggregates. The resulting formulas can be
used to enforce stability in various settings where the semantics is
no longer based on stable models themselves.

The normalization \citep{BJN20:acmtocl} of cardinality rules is
used to guide our intuitions about the intended generalization of
tight ordered completion. Besides this, to enable compact representations
of aggregates as propositional formulas, we extend the language
of difference logic by pseudo-Boolean constraints of the form
$\rg{c_1a_1}{+}{c_ma_m}\geq b$
where $\rg{a_1}{,}{a_m}$ are atoms, $\rg{c_1}{,}{c_m}$ their respective
integer coefficients, and $b$ an integer bound. Obviously, given an
interpretation $\pair{I}{\tau}$ in DL, we define
$\pair{I}{\tau}\models\rg{c_1a_1}{+}{c_ma_m}\geq b$,
iff $\sum_{I\models a_i}c_i\geq b$, since the truth values of
$\rg{a_1}{,}{a_m}$ are determined by $I$ independently of $\tau$.
Let us begin with an example that concentrates on a corner case ($l=1$
and $m=0$) of \eqref{eq:cardinality-rule} from Example
\ref{ex:positive-card}.

\begin{example}\label{ex:extract-TOC}
Recalling formulas
\eqref{eq:ex-def-of-atom}--\eqref{eq:ex-deny-gaps} from Example
\ref{ex:card-TOC}, we pull them back to the setting of the
original cardinality rule
$a\IF 1\leq\set{\rg{b_1}{,}{b_n}}$
where $\rg{b_1}{,}{b_n}$ depend recursively on the head $a$.
Based on a \emph{connecting} formula
$\rg{\app{a}{1}}{\lor}{\app{a}{n}}\lequiv\app{a}{}$
on the applicability of the $n$ rules in the normalization versus the
applicability of the original rule, we rewrite
\eqref{eq:ex-def-of-atom}--\eqref{eq:ex-deny-gaps} as follows:
\begin{eqnarray}
a\lequiv\app{a}{},
\label{eq:card-def-of-atom}
\\
\app{a}{}\lequiv(\rg{\dep{a}{b_1}}{+}{\dep{a}{b_n}}\geq 1),
\label{eq:card-def-of-body}
\\
\app{a}{}\limpl(\rg{\gap{a}{b_1}}{+}{\gap{a}{b_n}}<1),
\label{eq:card-deny-gaps}
\end{eqnarray}
where the new atoms $\rg{\dep{a}{b_1}}{,}{\dep{a}{b_n}}$ and
$\rg{\gap{a}{b_1}}{,}{\gap{a}{b_n}}$ are still to be interpreted
subject to formulas \eqref{eq:levels}--\eqref{eq:depdep} as
in the context of Example \ref{ex:card-TOC}.
\eofex
\end{example}

Note that the formula \eqref{eq:card-deny-gaps} expresses a
\emph{dynamic} check, i.e., it works for any subset $B$ of
$\eset{b_1}{b_n}$ of atoms \emph{true} and \emph{derived earlier} than
$a$. If the cardinality rule is applied (i.e., $|B|\geq 1$),
$\gap{a}{b_i}$ must be false for each $b_i\in B$, amounting to the
effect of the individual implications in \eqref{eq:ex-deny-gaps}.
The formulas in Examples \ref{ex:card-TOC} and \ref{ex:extract-TOC}
are based on different auxiliary atoms denoting the applicability of
rules. The connecting formula
$\rg{\app{a}{1}}{\lor}{\app{a}{n}}\lequiv\app{a}{}$
describes their intended semantic interconnection for propositional
interpretations $M$ and $N$, i.e.,
$M\models\rg{\app{a}{1}}{\lor}{\app{a}{n}}$ iff $N\models\app{a}{}$.
Because \eqref{eq:dep} and \eqref{eq:depdep} disconnect all integer
variables from the formulas under consideration, the following
proposition concentrates on the propositional parts of DL-interpretations
that are intended to satisfy TOC formulas in the end.

\begin{proposition}\label{prop:equisatisfiable}
The formulas \eqref{eq:ex-def-of-atom}--\eqref{eq:ex-deny-gaps} and
\eqref{eq:card-def-of-atom}--\eqref{eq:card-deny-gaps} constrain the
respective interpretations $\pair{M}{\tau}$ and $\pair{N}{\tau}$
equivalently, as conveyed by the satisfaction of the formula
$\rg{\app{a}{1}}{\lor}{\app{a}{n}}\lequiv\app{a}{}$.
\end{proposition}

\begin{proof}
Assuming the connecting formula makes formulas
\eqref{eq:ex-def-of-atom} and \eqref{eq:card-def-of-atom}
equivalent.
Formulas \eqref{eq:ex-def-of-atom} and \eqref{eq:ex-def-of-body} imply
$a\lequiv\rg{\dep{a}{b_1}}{\lor}{\dep{a}{b_n}}$ that is equivalent to
\eqref{eq:card-def-of-body} under \eqref{eq:card-def-of-atom}.
Finally, $\rg{\gap{a}{b_1}}{+}{\gap{a}{b_n}}<1$ is the same as
$\rg{\neg\gap{a}{b_1}}{\land}{\neg\gap{a}{b_n}}$.
These conditions are equally enforced by \eqref{eq:ex-deny-gaps} and
\eqref{eq:card-deny-gaps} when the connecting formula is satisfied.
\hfill
\end{proof}

Proposition \ref{prop:equisatisfiable} indicates that the formulas
\eqref{eq:ex-def-of-atom}--\eqref{eq:ex-deny-gaps} introduced for the
normalizing rules can be safely substituted by the formulas
\eqref{eq:card-def-of-atom}--\eqref{eq:card-deny-gaps} for the
original rule. In this way, the aggregated condition is restored
as a subformula in \eqref{eq:card-def-of-body} while its negation
incarnates in \eqref{eq:card-deny-gaps}.
Recall that the truth values of atoms $\gap{a}{b_i}$ are determined by
\eqref{eq:depdep}. If \eqref{eq:card-deny-gaps} were not satisfied by
$\pair{N}{\tau}$, at least one $\gap{a}{b_i}$ atom must be true, i.e.,
$N\models b_i$ and $\tau\models(x_a>x_{b_i}+1)$ assuming the
satisfaction of \eqref{eq:depdep}. Thus $b_i$ would be derived so
early that the derivation of $a$ is feasible before and the value of
$x_a$ could be decreased. Consequently, the joint effect of the
formulas \eqref{eq:card-def-of-body} and \eqref{eq:card-deny-gaps} is
that $x_a=\min\sel{x_{b_i}+1}{N\models b_i}$ holds which is in harmony
with the characterization of \citep{Janhunen06:jancl} when applied to
the normalizing rules $\rg{a\IF b_1\END}{}{a\IF b_n\END}$
Before addressing arbitrary cardinality rules, we draw the reader's
attention to the other extreme.

\begin{example}
When $l=n$ and $m=0$ in \eqref{eq:cardinality-rule},
the rule can be directly cast as a positive normal rule
$a\IF\rg{b_1}{\AND}{b_n}\END$
Still assuming that $\eset{b_1}{b_n}\subseteq\sccof{a}$, the TOC
formulas resulting from \eqref{eq:def-of-atom}--\eqref{eq:deny-gaps} are
$a\lequiv\app{a}{1}$,
$\app{a}{1}\lequiv\rg{\dep{a}{b_1}}{\land}{\dep{a}{b_n}}$, and
$\app{a}{1}\limpl\Lor_{1\leq i\leq n}\neg\gap{a}{b_i}$.
The corresponding aggregated formulas can be seen in the formulas
\eqref{eq:card-def-of-body} and \eqref{eq:card-deny-gaps}
if the bound $1$ is substituted by $n$.
The resulting \emph{strong} level ranking constraint ensures that at
least one body atom $b_i$ is derived \emph{just before} $a$ and
$x_a=\max\sel{x_{b_i}}{1\leq i\leq n}+1$.
\eofex
\end{example}

The preceding example reveals our plan when it comes to covering more
general lower bounds $1<l<n$ in \eqref{eq:cardinality-rule} still
pertaining to the positive case $m=0$ and
$\eset{b_1}{b_n}\subseteq\sccof{a}$. In the sequel, we write
$\subseteq_l$ to denote the $l$-subset relation restricted to subsets
of size $l$ \emph{exactly}. Due to monotonicity, the satisfaction of
the rule body $l\leq\eset{b_1}{b_n}$ essentially depends on the
$l$-subsets of $\eset{b_1}{b_n}$. Thus, the cardinality rule
\eqref{eq:cardinality-rule} with $m=0$ can be normalized by
introducing a positive rule $a\IF B$ for each
$B\subseteq_l\eset{b_1}{b_n}$. The number of such rules $\binom{n}{l}$
is at maximum when $l$ is $n$ halved.%
\footnote{Note that $\binom{n}{\lfloor n/2\rfloor}\geq 3^{\lfloor n/2\rfloor}$ from $n>4$ onward.}
In spite of exponential growth, the resulting normalization serves
the purpose of understanding the effect of $l$ on the required TOC
formulas.
To update equations
\eqref{eq:ex-def-of-atom}--\eqref{eq:ex-deny-gaps}
for this setting, we need a new atom $\app{a}{B}$ for every
$B\subseteq_l\eset{b_1}{b_n}$ to capture the individual
applicabilities of the respective positive rules $a\IF B$:
\begin{eqnarray}
a & \lequiv & \Lor_{B\subseteq_l\eset{b_1}{b_n}}\app{a}{B},
\label{eq:ex-def-of-atom2}
\\
\app{a}{B} & \lequiv & \Land_{b\in B}\dep{a}{b} ~~
(\text{for }B\subseteq_l\eset{b_1}{b_n}),
\label{eq:ex-def-of-body2}
\\
\app{a}{B} & \limpl & \Lor_{b\in B}\neg\gap{a}{b} ~~
(\text{for }B\subseteq_l\eset{b_1}{b_n}).
\label{eq:ex-deny-gaps2}
\end{eqnarray}
The connecting formula
$\Lor_{B\subseteq_k\eset{b_1}{b_n}}\app{a}{B}\lequiv\app{a}{}$
links the above back to the original rule
$a\IF l\leq\eset{b_1}{b_n}$
suggesting the revisions of
\eqref{eq:card-def-of-atom}--\eqref{eq:card-deny-gaps}
for any lower bound $1\leq l\leq n$:
\begin{eqnarray}
a & \lequiv & \app{a}{},
\label{eq:card-def-of-atom2}
\\
\app{a}{} & \lequiv & (\rg{\dep{a}{b_1}}{+}{\dep{a}{b_n}}\geq l),
\label{eq:card-def-of-body2}
\\
\app{a}{} & \limpl & (\rg{\gap{a}{b_1}}{+}{\gap{a}{b_n}}<l).
\label{eq:card-deny-gaps2}
\end{eqnarray}
Most importantly, the length of the formulas
\eqref{eq:card-def-of-atom2}--\eqref{eq:card-deny-gaps2}
stays linear in $n$ in contrast with their alternatives
\eqref{eq:ex-def-of-atom2}--\eqref{eq:ex-deny-gaps2}
based on $l$-subsets. The aggregate-based formulation covers all
$l$-subsets of $\eset{b_1}{b_n}$ and their supersets that also satisfy
the body of \eqref{eq:cardinality-rule} with $m=0$ by monotonicity.

\begin{proposition}\label{prop:equisatisfiable2}
The formulas \eqref{eq:ex-def-of-atom2}--\eqref{eq:ex-deny-gaps2} and
\eqref{eq:card-def-of-atom}--\eqref{eq:card-deny-gaps2} constrain the
respective interpretations $\pair{M}{\tau}$ and $\pair{N}{\tau}$
equivalently, as conveyed by the satisfaction of the formula
$\Lor_{B\subseteq_l\eset{b_1}{b_n}}\app{a}{B}\lequiv\app{a}{}$.
\end{proposition}

The proof is similar to that of Proposition \ref{prop:equisatisfiable}
and amounts to showing that the (big) disjunctive formula
$\Lor_{B\subseteq_l\eset{b_1}{b_n}}\Land_{b\in B}\dep{a}{b}$
is expressible as $(\rg{\dep{a}{b_1}}{+}{\dep{a}{b_n}}\geq l)$.
Similar aggregation is achieved in \eqref{eq:card-deny-gaps2} in terms
of $\rg{\gap{a}{b_1}}{,}{\gap{a}{b_1}}$. An important observation is
that $\pair{N}{\tau}\not\models$ \eqref{eq:card-deny-gaps2} if and
only if $N\models\app{a}{}$ and $\exists B\subseteq_l\eset{b_1}{b_n}$
such that for each $b\in B$, $N\models b$, $N\models\gap{a}{b}$, and
$\tau\models(x_a>x_b+1)$. Since $B$ reaches the bound $l$ the
value of $x_a$ could be decreased to $\max\sel{\tau(x_b)}{b\in B}+1$
or even further if $|B|>l$.
Thus the satisfaction of \eqref{eq:card-deny-gaps2} means that no
$B\subseteq_l(\eset{b_1}{b_n}\isect N)$ of true atoms could be used to
decrease the value of $x_a$. The net effect is that $x_a$ has the
critical minimum value. Since \eqref{eq:card-def-of-body2} is also
satisfied there are at least $l$ true atoms derived before $a$ but
sufficiently many of them are derived \emph{just before} $a$. For
those atoms $b$ we have $x_a=x_b+1$, $\dep{a}{b}$ true, but
$\gap{a}{b}$ false!

\begin{example}\label{ex:card-critical}
Consider the rule $a\IF 2\leq\set{b_1,b_2,b_3,b_4}$ in the context
of a model $\pair{N}{\tau}$ where $N=\set{a,b_1,b_3,b_4,\app{a}{}}$,
$\tau(x_{b_1})=\tau(x_{b_4})=2$, $\tau(x_{b_3})=1$, and $\tau(x_{b_2})=6$
by default as $N\not\models b_2$, see \eqref{eq:levels}.
Now the rule body is satisfied by three $2$-subsets
$B_1=\set{b_1,b_3}$, $B_2=\set{b_1,b_4}$, and $B_3=\set{b_3,b_4}$
justifying the level rank $\tau(x_a)=3$, since
$\max\sel{\tau(x_b)}{b\in B_i}=2$
for each $1\leq i\leq 3$. We have $\gap{b_i}{a}$ true only for $i=3$
and thus \eqref{eq:card-deny-gaps2} is respected for $l=2$. But, if
$\tau(x_a)=4$ had been alternatively set, the count of such $b_i$'s
would be $3$, falsifying \eqref{eq:card-deny-gaps2}.
\eofex
\end{example}

\subsection{Weights}
\label{section:weights}

So far, we have established relatively general form of TOC formulas
\eqref{eq:card-def-of-atom2}--\eqref{eq:card-deny-gaps2} that cover
cardinality rules \eqref{eq:cardinality-rule} when $m=0$. Before
addressing negative body conditions ($m>0$) and settings where SCCs
play a major role, we take the weights of literals into consideration
as already present in weight rules \eqref{eq:weight-rule} when $m=0$.
Consequently, we have to substitute $l$-subsets used in
\eqref{eq:ex-def-of-atom2}--\eqref{eq:ex-deny-gaps2}
by \emph{weighted} subsets
$B\subseteq_w\eset{b_1=\wght{b_1}}{b_n=\wght{b_n}}$.
Such a subset $B$ can be formally defined in terms of the condition
$\wsum{B}{\eset{b_1=\wght{b_1}}{b_n=\wght{b_n}}}\geq w$
from Section \ref{section:preliminaries}.
It is clear by monotonicity that if
$B\subseteq_w\eset{b_1=\wght{b_1}}{b_n=\wght{b_n}}$, then
$B'\subseteq_w\eset{b_1=\wght{b_1}}{b_n=\wght{b_n}}$ for every $B'$
with $B\subseteq B'\subseteq\eset{b_1}{b_n}$. A weighted set
$B\subseteq_w\eset{b_1=\wght{b_1}}{b_n=\wght{b_n}}$ is defined to be
$\subseteq$-minimal with respect to $w$, if for no $B'\subset B$,
$B'\subseteq_w\eset{b_1=\wght{b_1}}{b_n=\wght{b_n}}$.
We use $\subseteq_w^{\min}$ to indicate such $\subseteq$-minimal weighted
subsets of $\eset{b_1}{b_n}$.
Assuming orthogonal generalizations of
\eqref{eq:ex-def-of-atom2}--\eqref{eq:ex-deny-gaps2}
for a \emph{positive} weight rule \eqref{eq:weight-rule} and the
weighted subsets $B\subseteq_w^{\min}\eset{b_1=\wght{b_1}}{b_n=\wght{b_n}}$
of its body, we rather incorporate weights into the formulas
\eqref{eq:card-def-of-atom2}--\eqref{eq:card-deny-gaps2}
as follows:

\begin{eqnarray}
a & \lequiv & \app{a}{},
\label{eq:weight-def-of-atom}
\\
\app{a}{} & \lequiv &
(\rg{\wght{b_1}\times\dep{a}{b_1}}{+}{\wght{b_n}\times\dep{a}{b_n}}\geq w),
\label{eq:weigth-def-of-body}
\\
\app{a}{} & \limpl &
(\rg{\wght{b_1}\times\gap{a}{b_1}}{+}{\wght{b_n}\times\gap{a}{b_n}}<w).
\label{eq:weight-deny-gaps}
\end{eqnarray}

\begin{proposition}\label{prop:equisatisfiable3}
The formulas \eqref{eq:ex-def-of-atom2}--\eqref{eq:ex-deny-gaps2} revised for
weighted subsets $B$ subject to the bound $w$ and the formulas
\eqref{eq:weight-def-of-atom}--\eqref{eq:weight-deny-gaps} constrain the
respective interpretations $\pair{M}{\tau}$ and $\pair{N}{\tau}$
equivalently, as conveyed by the satisfaction of the equivalence
$\Lor_{B\subseteq_w^{\min}\eset{b_1=\wght{b_1}}{b_n=\wght{b_n}}}\app{a}{B}\lequiv\app{a}{}$.
\end{proposition}

\begin{proof}
Due to high similarity with respect to Proposition \ref{prop:equisatisfiable2},
we just point out the equivalence of the formulas
$\Lor_{B\subseteq_w^{\min}\eset{b_1=\wght{b_1}}{b_n=\wght{b_n}}}
 \Land_{b\in B}\dep{a}{b}$
and
$(\rg{\wght{b_1}\times\dep{a}{b_1}}{+}{\wght{b_n}\times\dep{a}{b_n}}\geq w)$.
The equivalence involving $\rg{\gap{a}{b_1}}{,}{\gap{a}{b_n}}$ is analogous
but negated.
\hfill
\end{proof}

Again, $\pair{N}{\tau}\not\models$ \eqref{eq:weight-deny-gaps}
implies $N\models\app{a}{}$ and for some
$B\subseteq_w^{\min}\eset{b_1=\wght{b_1}}{b_n=\wght{b_n}}$,
$N\models B$ and for every $b\in B$, $\tau\models(x_a>x_b+1)$.
Then, the value of $x_a$ could be decreased to
$\max\sel{\tau(x_b)}{b\in B}+1$. Thus the formula \eqref{eq:weight-deny-gaps}
makes $\tau(x_a)$ minimal as before.

\begin{example}\label{ex:positive-weight}
Let us consider a positive weight rule
$a\IF 7\leq\set{b_1=7,b_2=5,b_3=3,b_4=2,b_5=1}$
in the context of a model $\pair{N}{\tau}$ where $N$ sets all
body atoms $\rg{b_1}{,}{b_5}$ true and $\tau$ the level numbers
\begin{center}
$\tau(x_a)=5$, ~~
$\tau(x_{b_1})=5$, ~~
$\tau(x_{b_2})=4$, ~~
$\tau(x_{b_3})=3$, ~~
$\tau(x_{b_4})=2$, ~~
$\tau(x_{b_5})=1$.
\end{center}
The $\subseteq$-minimal satisfiers of the body are
$B_1=\set{b_1}$, $B_2=\set{b_2,b_3}$, and $B_3=\set{b_2,b_4}$.
The only atom in $B_1$ has a weight that reaches the bound $7$ alone,
but it is derived too late to affect the derivation of $a$. Both $B_2$
and $B_3$ yield the same value $\max\sel{\tau(x_b)}{b\in B_i}=4$ and
hence justify the one higher value $5$ assigned to
$x_a$. Interestingly, there is also an atom $b_5$ that is derived
early, but whose weight is irrelevant for satisfying the rule body nor
deriving $a$ any earlier. In fact, this weighted atom could be safely
deleted from the rule (under strong equivalence).

As regards the satisfaction of \eqref{eq:weight-deny-gaps}, the relevant
body atoms are $b_3$, $b_4$, and $b_5$, for which the atom
$\gap{a}{b_i}$ is made true by \eqref{eq:depdep}. The
sum of the respective weights $3+2+1$ is less than $7$.

Also, note that the level numbers assigned by $\tau$ to
$\rg{b_1}{,}{b_5}$ can be easily arranged with positive rules,
e.g., by using the chain of rules:
$b_1\IF b_2\END$
$b_2\IF b_3\END$
$b_3\IF b_4\END$
$b_4\IF b_5\END$
$b_5\END$
Given the respective program $P$, the operator $\imt{P}$ should
be applied $5$ times to make $a$ true.
\eofex
\end{example}

\subsection{Negative Conditions}
\label{section:negative}

Negative body conditions form the missing pieces when it comes to
fully covering WCPs with level ranking constraints as embedded in
tight ordered completion. To this end, our strategy is based on
rewriting and ideas used in \citep{BJN20:acmtocl} where the
correctness of normalization is first shown for positive programs and
then generalized for programs with negation.
In a nutshell, negative literals in \eqref{eq:weight-rule} can be
replaced by new atoms $\rg{\compl{c_1}}{,}{\compl{c_m}}$ that
respectively denote that $\rg{c_1}{,}{c_m}$ cannot be derived.
These atoms are subsequently defined by (atomic) normal rules
$\rg{\compl{c_1}\IF\naf c_1\END}{~~}{\compl{c_m}\IF\naf c_m\END}$
The outcome is a set of rules that is visibly strongly equivalent with
the original weight rule \eqref{eq:weight-rule}. The completions of
$\rg{\compl{c_1}}{,}{\compl{c_m}}$
are $\rg{\compl{c_1}\lequiv\neg c_1}{,}{\compl{c_m}\lequiv\neg c_m}$,
enabling the substitution of
$\rg{\compl{c_1}}{,}{\compl{c_m}}$
by $\rg{\neg c_1}{,}{\neg c_m}$ in any formulas of interest. In this
way, $\rg{\compl{c_1}}{,}{\compl{c_m}}$ can be readily forgotten under
classical semantics.
The transformation described above leaves the SCCs of the program
intact, because positive dependencies are not affected. Thus, besides
taking care of negative body conditions, our next rewriting step
recalls the scope $S\subseteq\sig{P}$ from
\eqref{eq:def-of-body}--\eqref{eq:reset-ext}:
we present TOC formulas to cover WCPs split into modules based on
SCCs. We say that a WCP is \emph{pure} if it contains weight rules
\eqref{eq:weight-rule} only.

\begin{definition}\label{def:TOC-of-WCP}
Let $P$ be a pure WCP and $S\subseteq\sig{P}$ an SCC of $P$ used as
the scope of completion.
The \emph{tight ordered completion} of $P$ relative to $S$, denoted
$\tr{TOC}{S}{P}$, consists of the formulas listed below:

\begin{itemize}
\item
If $|S|>1$, then for each $a\in S$:
\begin{enumerate}
\item the formulas \eqref{eq:levels};
\item the formulas \eqref{eq:dep} and \eqref{eq:depdep} for each $b\in S$ such that $a\pdep{P}b$; plus
\item the following formulas based on the definition $\defof{a}{P}=\eset{r_1}{r_k}$in the program $P$:
\begin{eqnarray}
a & \lequiv & \rg{\app{a}{1}}{\lor}{\app{a}{k}},
\label{eq:WCP-def-of-atom}
\\
\app{a}{i} & \lequiv & \int{a}{i}\lor\ext{a}{i},
\label{eq:WCP-int-or-ext}
\\
\int{a}{i} & \lequiv &
\textstyle
(\sum_{b\in\pbody{r_i}\isect S}(\wght{b}\times\dep{a}{b})
 +\sum_{b\in\pbody{r_i}\setminus S}(\wght{b}\times b) \nonumber \\
           &         &
\textstyle
\phantom{(}-\sum_{c\in\nbody{r_i}}(\wght{c}\times c)
 \geq \wght{r_i}-w_i),
\label{eq:WCP-def-of-body}
\\
\int{a}{i} & \limpl &
\textstyle
(\sum_{b\in\pbody{r_i}\isect S}(\wght{b}\times\gap{a}{b})
 +\sum_{b\in\pbody{r_i}\setminus S}(\wght{b}\times b) \nonumber \\
           &        &
\textstyle
\phantom{(}-\sum_{c\in\nbody{r_i}}(\wght{c}\times c)
 < \wght{r_i}-w_i)\, \lor\, \ext{a}{i},
\label{eq:WCP-deny-gaps}
\\
\ext{a}{i} & \lequiv &
\textstyle
(\sum_{b\in\pbody{r_i}\setminus S}(\wght{b}\times b)-\sum_{c\in\nbody{r_i}}(\wght{c}\times c) \nonumber \\
           &        &
\phantom{(}\geq \wght{r_i}-w_i),
\label{eq:WCP-def-of-ext}
\\
\ext{a}{i} & \limpl & (x_a\leq 1),
\label{eq:WCP-reset-ext}
\end{eqnarray}
where $w_i=\sum_{c\in\nbody{r_i}}\wght{c}$ is the adjustment to the bound
$\wght{r_i}$ of $r_i$.
\end{enumerate}

\item
If $|S|=1$ and $S=\set{a}=\sccof{a}$, then
\eqref{eq:WCP-def-of-atom}--\eqref{eq:WCP-def-of-ext}
are replaced by the standard completion
\begin{equation}
\app{a}{i}\lequiv
(\sum_{b\in\pbody{r_i}}(\wght{b}\times b)-\sum_{c\in\nbody{r_i}}(\wght{c}\times c)
 \geq \wght{r_i}-w_i).
\label{eq:WCP-trivial-body}
\end{equation}
\end{itemize}
\end{definition}

In the equations of Definition \ref{def:TOC-of-WCP}, the treatment of
negative literals occurring in a defining rule $r_i$ is justified by
their contribution
$\sum_{c\in\pbody{r_i}}(\wght{c}\times(1-c))=
w_i-\sum_{c\in\pbody{r_i}}(\wght{c}\times c)$
toward the bound $\wght{r_i}$.
Weight rules can also create external support in more flexible ways,
i.e., if the bound $\wght{r_i}$ can be reached by satisfying any
positive body conditions outside the SCC $S$ in question or any
negative body conditions. Moreover, a single weight rule $r_i$ may
justify the head $a$ either internally or externally as formalized by
\eqref{eq:WCP-int-or-ext}, which is different from the case of normal
rules. Formulas \eqref{eq:WCP-def-of-body} and \eqref{eq:WCP-def-of-ext}
capture this distinction. Note that $\ext{a}{i}$ implies $\int{a}{i}$.
The constraint \eqref{eq:WCP-deny-gaps} generalizes \eqref{eq:deny-gaps}
while \eqref{eq:WCP-reset-ext} is the analog of \eqref{eq:reset-ext}.
The consequent of \eqref{eq:WCP-deny-gaps} is weakened by the condition
$\ext{a}{i}$, since the other disjunct is falsified if $r_i$ provides
external support, $x_a=1$ holds, and all atoms $\dep{a}{\cdot}$
and $\gap{a}{\cdot}$ associated with $a$ are falsified. The
reader may have noticed that the formulas concerning $\int{a}{i}$
and $\ext{a}{i}$ share a potentially large subexpression
$s_i=\sum_{b\in\pbody{r_i}\setminus S}(\wght{b}\times b)
    -\sum_{c\in\nbody{r_i}}(\wght{c}\times c)$.
Certain back-end formalisms, such as DL and MIP, enable the representation
of this expression only once using an integer variable.

The following theorem states the correctness of $\tr{TOC}{}{P}$
obtained as the union of $\tr{TOC}{S}{P}$ for the SCCs $S$ of $P$. The
claimed one-to-one correspondence, however, must take into account the
fact that a satisfying assignment $\pair{M}{\tau}$ in DL can be
replicated into infinitely many copies by substituting $\tau$ in
$\pair{M}{\tau}$ by a function $\tau'(x)=\tau(x)+k$ for any
$k\in\integers$. The remedy is to introduce a special variable $z$
which is assumed to hold $0$ as its value and the values of the other
variables are set relative to the value of $z$. The current difference
constraints mentioning only one variable must be rewritten using $z$.
For instance, $1\leq x_a\leq |S|+1$ in \eqref{eq:levels} is expressed
by the conjunction of $z-x_a\leq -1$ and $x_a-z\leq|S|+1$, and this
is how the variable $z$ gets introduced.

\begin{theorem}\label{theorem:WCP-correctness}
Let $P$ be a WCP with SCCs $\rg{S_1}{,}{S_q}$ and $\rg{P_{S_1}}{,}{P_{S_q}}$
the respective modules of $P$.
Then $P$ and the set of formulas $F=\Union_{j=1}^q\tr{TOC}{S_j}{P}$
are visibly equivalent (up to assigning $z=0$).
\end{theorem}

\section{Generalizations Toward Convex Aggregates}
\label{section:convex-aggregates}

Weight rules \eqref{eq:weight-rule} can be generalized by introducing
upper bounds $u$ besides lower bounds $l$:
\begin{equation}
\label{eq:lower-and-upper-bound}
a\IF l\leq\set{\rg{b_1=\wght{b_1}}{,}{b_n=\wght{b_n}},
               \rg{\naf c_1=\wght{c_1}}{,}{\naf c_m=\wght{c_m}}}\leq u\END
\end{equation}
This gives a rise to a \emph{convex} condition which is easier to
explain for positive rules ($m=0$). If the condition can be satisfied
by setting the atoms of $B_1\subseteq\eset{b_1}{b_n}$ true and the
same holds for a superset $B_2\subseteq\eset{b_1}{b_n}$ of $B_1$, then
every intermediate set $B'$ such that $B_1\subseteq B'\subseteq B_2$
satisfies the condition, too. It is easy to check that this is a property
of \eqref{eq:lower-and-upper-bound}, since $B_1\subseteq B_2$ implies
$\wsum{B_1}{\eset{b_1=\wght{b_1}}{b_n=\wght{b_n}}}\leq
 \wsum{B_2}{\eset{b_1=\wght{b_1}}{b_n=\wght{b_n}}}$
in general.
However, the bounds play a role here: upper bounds jeopardize
monotonicity in general but convexity is still guaranteed.
Thus, we use WCPs based on \eqref{eq:lower-and-upper-bound} to
understand the role of level ranking constraints in the context of
convex aggregates. The effect of negative literals is anti-monotonic,
but their semantics is determined by the reduct as
usual. \tconly{Simons et al.}~\cite{SNS02:aij} present a
transformation that checks the upper bound of
\eqref{eq:lower-and-upper-bound} with another weight rule.
The rules below adopt this idea but using a constraint and new atoms
that are in line with \eqref{eq:WCP-trivial-body} and $r_i\in\defof{a}{P}$:
\begin{eqnarray}
\app{a}{i}
& \IF & l\leq\set{\rg{b_1=\wght{b_1}}{,}{b_n=\wght{b_n}},
                  \rg{\naf c_1=\wght{c_1}}{,}{\naf c_m=\wght{c_m}}}\END
\label{eq:lower-bound}
\\
\vub{a}{i}
& \IF & u+1\leq\set{\rg{b_1=\wght{b_1}}{,}{b_n=\wght{b_n}},
                    \rg{\naf c_1=\wght{c_1}}{,}{\naf c_m=\wght{c_m}}}\END
\label{eq:upper-bound}
\\
& \IF & \app{a}{i}\AND\vub{a}{i}\END
\label{eq:check-bounds}
\end{eqnarray}
For the moment, this relaxes the notion of applicability for $r_i$ but
the constraint \eqref{eq:check-bounds} makes sure that the upper bound
is not \emph{violated}. Since no atom depends positively on
$\vub{a}{i}$, the resulting TOC formula is analogous to
\eqref{eq:WCP-trivial-body} with $\vub{a}{i}$ as its head. The
constraint \eqref{eq:check-bounds} can be expressed as
$\neg(\app{a}{i}\land\vub{a}{i})$. The net effect is that if
$\app{a}{i}$ is true, then
$(\sum_{b\in\pbody{r_i}}(\wght{b}\times b)
  -\sum_{c\in\nbody{r_i}}(\wght{c}\times c)
  \leq u-w_i)$
must hold. Thus \eqref{eq:WCP-def-of-body}--\eqref{eq:WCP-def-of-ext}
can be revised for \eqref{eq:lower-and-upper-bound} by replacing the
previous lower bound $\wght{r_i}$ with $l$ and by incorporating
$u-w_i$
into \eqref{eq:WCP-def-of-body} and \eqref{eq:WCP-def-of-ext} as upper
bounds; either as two pseudo-Boolean constraints or a combined one
with two bounds. The respective upper bound does not play a role in
\eqref{eq:WCP-deny-gaps} that concerns the criticality of the lower
bound and this check does not interfere with the satisfaction of the
upper bound due to convexity.

\subsection{Abstract Constraint Atoms}

Based on the preceding analysis of weight rules, we will rephrase our
approach for an arbitrary convex aggregate $\aggr{}{B}$ that takes a
set of (body) atoms $B$ as input and accepts a certain subset
$\cal{S}$ of the powerset $\pset{B}$ by evaluating to true. This set
must satisfy the convexity condition, i.e., if $S_1,S_2\in\cal{S}$
then $S\in\cal{S}$ for each intermediate set $S_1\subseteq S\subseteq
S_2$, too. Moreover, we let $\aggr{*}{B}$ stand for the
\emph{monotonic} (upward) \emph{closure} of $\aggr{}{B}$ based on the
signature $B$.  The set of satisfiers $\mathcal{S}\subseteq\pset{B}$
of the latter is extended to a set of satisfiers
$\mathcal{S}^*=\sel{S}{\exists S'\in\mathcal{S},~S'\subseteq S\subseteq B}$
for the former. The syntax of logic programs can be extended by
introducing aggregated conditions as rule bodies in analogy to
\eqref{eq:cardinality-rule} and \eqref{eq:weight-rule}:
\begin{equation}
a\IF\aggr{}{\rg{b_1}{,}{b_n},\rg{\naf c_1}{,}{\naf c_m}}\END
\label{eq:aggr-rule}
\end{equation}
Let $P$ be a logic program consisting of rules of the form
\eqref{eq:aggr-rule} and $M\subseteq\sig{P}$ a \emph{model} of $P$
that satisfies all rules \eqref{eq:aggr-rule} in the standard sense,
i.e., if the body is satisfied by $M$, then the head $a\in M$.  The
reduct $\GLred{P}{M}$ can be formed by including a positive rule
$a\IF\aggr{*}{\rg{b_1}{,}{b_m},\rg{\csubst{c_1\in M}{\false}{\true}}{,}{\csubst{c_m\in M}{\false}{\true}}}$
for each \eqref{eq:aggr-rule} such that
$M\models\aggr{}{\rg{b_1}{,}{b_n},\rg{\naf c_1}{,}{\naf c_m}}$.
In the above, we exploit \emph{conditional substitutions}
$\csubst{c}{v}{u}$ by the value $v$ if the condition $c$ is true and
the value $u$ otherwise. Thus, given a set of input atoms
$I\subseteq\sig{P}\setminus\heads{P}$,
we can calculate the least model of $\GLred{P}{M}$ and assign level
ranks $\level{a}=i$ of atoms $a\in\heads{P}$ based on the membership
$a\in\iter{\imt{\GLred{P}{M}}}{i}{I}
     \setminus\iter{\imt{\GLred{P}{M}}}{i-1}{I}$
as earlier. Assuming a defining rule $r_i\in\defof{a}{P}$ for an
atom $a\in\heads{P}$ and the scope $S=\sccof{a}$, we generalize
\eqref{eq:WCP-def-of-body}--\eqref{eq:WCP-def-of-ext}
as follows:
\begin{eqnarray}
\int{a}{i} & \!\!\lequiv\!\! &
\aggr{}{\rg{\csubst{b_1\in S}{\dep{a}{b_1}}{b_1}}{,}{\csubst{b_n\in S}{\dep{a}{b_n}}{b_n}},\rg{\neg c_1}{,}{\neg c_m}},
\label{eq:aggr-def-of-body}
\\
\int{a}{i} & \!\!\limpl\!\! &
\neg
\aggr{*}{\rg{\csubst{b_1\in S}{\gap{a}{b_1}}{b_1}}{,}{\csubst{b_n\in S}{\gap{a}{b_n}}{b_n}},\rg{\neg c_1}{,}{\neg c_m}}
\nonumber
\\
           &        &
\lor\,\ext{a}{i},
\label{eq:aggr-deny-gaps}
\\
\ext{a}{i} & \!\!\lequiv\!\! &
\aggr{}{\rg{\csubst{b_1\in S}{\false}{b_1}}{,}{\csubst{b_n\in S}{\false}{b_n}},\rg{\neg c_1}{,}{\neg c_m}}.
\label{eq:aggr-def-of-ext}
\end{eqnarray}
The formulas above assume that the aggregate $\aggr{}{\cdot}$ can be
expressed as a propositional formula, or the target language has
sufficient syntactic primitives available such as pseudo-Boolean
constraints. Conditional substitutions are necessary, because we have
not made any assumptions about ordering the arguments to
$\aggr{}{\cdot}$. For instance, if the weight constraint from
Example~\ref{ex:positive-weight} were abstracted as
$\aggr{}{b_1,b_2,b_3,b_4,b_5}$, the meaning of
$\aggr{}{b_5,b_4,b_3,b_2,b_1}$ would become different due to
asymmetric weights.
As in the case of pure WCPs, the TOC formula
\eqref{eq:aggr-def-of-body} takes care of the (potential) internal
support delivered by $r_i$ and the support is ordered by the conjoined
atoms $\dep{a}{b_i}$ so that the stability of models can be
guaranteed. The formula \eqref{eq:aggr-def-of-ext} for external
support is analogous except that the contribution of positive body
conditions from the same component is hindered by substituting them by
$\false$. By the side of these formulas, the strong ranking constraint
\eqref{eq:aggr-deny-gaps} ensures that $\aggr{}{\cdot}$ could not be
satisfied with arguments derived way before $a$.

\section{Conclusions}
\label{section:conclusions}

In this work, we investigate potential generalizations of
\emph{level ranking constraints} \citep{Niemela08:amai} for extended
rule types involving aggregates. It turns out that the structure
of aggregates can be preserved if the back-end language offers
analogous syntactic primitives for expressing such conditions.
As a consequence, it is not necessary to normalize rules before the
introduction of ranking constraints in a form or another, depending on
the target formalism and the back-end solver to be used for actual
computations. The particular novelty of our approach lies in the
generalization of \emph{strong} level ranking constraints for WCPs and
other monotone/convex aggregates. Using them level rankings can be
made unique which is desirable, e.g., when counting answer sets.
A further by-product is that any WCP can be translated into a
\emph{tight} WCP if the formulas in $\tr{TOC}{}{P}$ are expressed with
rules rather than formulas, also justifying ``\emph{tight}'' as part of
TOC.

Although the results of this article are theoretical by nature, they
enable new kinds of strategies when it comes to implementing the
search of stable models using existing solver technology for SAT, SMT,
and MIP. E.g., the presented TOC formulas offer a common ground for
the translators in the \system{lp2*} family \citep{Janhunen18:ki}.
We leave \emph{non-convex} aggregates as future work for two main
reasons. First, there is no consensus about their semantics when
recursive definitions are enabled \citep{AFG23:tplp}. The ASP-core-2
language standard assumes the \emph{stratified} setting only whereas
the Clingo system implements one particular semantics for recursive
non-convex aggregates \citep{GKS16:corr}.
Second, there is also evidence \citep{AFG15:tplp} that the removal of
non-convex aggregates tends to produce disjunctive rules which go
beyond level rankings in the first place.
One potential solution is provided by the \emph{decomposition}
of non-convex aggregates into their maximal convex regions,
cf.~\citep{Janhunen10:iclp,LGJNY11:lpnmr}.

\paragraph{Acknowledgments}
This research is partially supported by the Academy of Finland projects
AI-ROT (\#335718), XAILOG (\#345633), and ETAIROS (\#352441).

%------------------------------------------------------------------------------

\end{document}